\documentclass[journal]{IEEEtran}

\newtheorem{remark}{Remark}
\newtheorem{lemma}{Lemma}
\newtheorem{theorem}{Theorem}
\newtheorem{definition}{Definition}

\usepackage{cite}

\usepackage{graphicx}
\usepackage[caption=false,font=footnotesize]{subfig}
\usepackage{float}

\usepackage[cmex10]{amsmath}

\usepackage{algorithm}
\usepackage{algorithmic}

\usepackage{float}

\begin{document}

\title{Source Delay in Mobile Ad Hoc Networks}

\author{Juntao~Gao,
				Yulong~Shen,~\IEEEmembership{Member,~IEEE,}
				Xiaohong~Jiang,~\IEEEmembership{Senior~Member,~IEEE}
        and~Jie~Li,~\IEEEmembership{Senior~Member,~IEEE}
}

\maketitle

\begin{abstract}
Source delay, the time a packet experiences in its source node, serves as a fundamental quantity for delay performance analysis in networks. 
However, the source delay performance in highly dynamic mobile ad hoc networks (MANETs) is still largely unknown by now.
This paper studies the source delay in MANETs based on a general packet dispatching scheme with dispatch limit $f$ (PD-$f$ for short), where a same packet will be dispatched out up to $f$ times by its source node such that packet dispatching process can be flexibly controlled through a proper setting of $f$.
We first apply the Quasi-Birth-and-Death (QBD) theory to develop a theoretical framework to capture the complex packet dispatching process in PD-$f$ MANETs. 
With the help of the theoretical framework, we then derive the cumulative distribution function as well as mean and variance of the source delay in such networks.
Finally, extensive simulation and theoretical results are provided to validate our source delay analysis and illustrate how  source delay in MANETs are related to network parameters.

\end{abstract}


%
\IEEEpeerreviewmaketitle

\section{Introduction}

\IEEEPARstart{M}{obile} ad hoc networks (MANETs) represent a class of self-configuring and infrastructureless networks with mobile nodes.
As MANETs can be rapidly deployed, reconfigured and extended at low cost, they are highly appealing for a lot of critical applications, like disaster relief, emergency rescue, battle field communications, environment monitoring, etc \cite{Andrews_CM08,Goldsmith_CM11}. To facilitate the application of MANETs in providing delay guaranteed services in above applications, understanding the delay performance of these networks is of fundamental importance \cite{Hanzo_Survey07,Chen_Network07}.

Source delay, the time a packet experiences in its source node, is an indispensable behavior in any network. Since the source delay is a delay quantity common to all MANETs, it serves as a fundamental quantity for delay performance analysis in MANETs. 
For MANETs without packet redundancy \cite{Grossglauser_TON02,Neely_IT05} and with one-time broadcast based packet redundancy \cite{Gao_WiOpt13}, the source delay actually serves as a practical lower bound for and thus constitutes an essential part of overall delay in those networks.
The source delay is also an indicator of packet lifetime, i.e., the maximum time a packet could stay in a network; in particular, it lower bounds the lifetime of a packet and thus serves as a crucial performance metric for MANETs with packet lifetime constraint.




Despite much research activity on delay performance analysis in MANETs (see section~\ref{section:related_works} for related works), the source delay performance of such networks is still largely unknown by now. 
The source delay analysis in highly dynamic MANETs is challenging, since it involves not only complex network dynamics like node mobility, but also issues related to medium contention, interference, packet generating and packet dispatching.  
This paper devotes to a thorough study on the source delay in MANETs under the practical scenario of limited buffer size and also a general packet dispatching scheme with dispatch limit $f$ (PD-$f$ for short). With the PD-$f$ scheme, a same packet will be dispatched out up to $f$ times by its source node such that packet dispatching process can be flexibly controlled through a proper setting of $f$.
The main contributions of this paper are summarized as follows.

\begin{itemize}

\item We first apply the Quasi-Birth-and-Death (QBD) theory to develop a theoretical framework to capture the complex packet dispatching process in a PD-$f$ MANET. The theoretical framework is powerful in the sense it enables complex network dynamics to be incorporated into source delay analysis, like node mobility, medium contention, interference, packet transmitting and packet generating processes. 


\item With the help of the theoretical framework, we then derive the cumulative distribution function (CDF) as well as mean and variance of the source delay in the considered MANET. By setting $f=1$ in a PD-$f$ MANET, the corresponding source delay actually serves as a lower bound for overall delay.

\item Extensive simulation results are provided to validate our theoretical framework and the source delay models. Based on the theoretical source delay models, we further demonstrate how source delay in MANETs is related to network parameters, such as packet dispatch limit, buffer size and packet dispatch probability. 


\end{itemize}

The rest of this paper is organized as follows. Section~\ref{section:network_models} introduces preliminaries involved in this source delay study. 
A QBD based theoretical framework is developed to model the source delay in Section~\ref{section:framework}. We derive in Section~\ref{section:delay} the CDF as well as mean and variance of the source delay. 
Simulation/numerical studies and the corresponding discussions are provided in Section~\ref{section:numerical}. Finally, we introduce related works regarding delay performance analysis in MANETs in Section~\ref{section:related_works} and conclude the  paper in Section~\ref{section:conclusion}.

\section{Preliminaries} \label{section:network_models}
In this section, we introduce the basic system models, the Medium Access Control (MAC) protocol and the packet dispatching scheme involved in this study.

\subsection{System Models}

\begin{figure}[!t]
    \centering
    {
    \subfloat[A snapshot of a cell partitioned MANET.]
    {\includegraphics[width=2in]{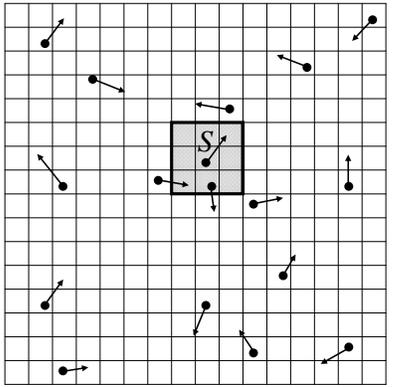} \label{fig:network_partition}}
   	\hfill
    \subfloat[Illustration of MAC-EC protocol.]
    {\includegraphics[width=2in]{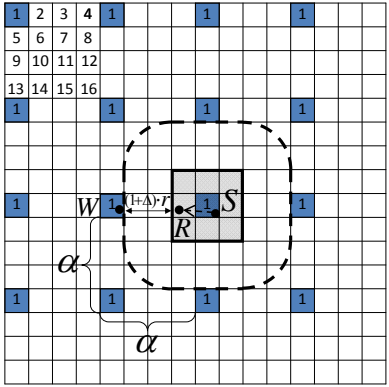} \label{fig:network_mac} }
    }
    \caption{An example of a cell partitioned MANET with a MAC protocol.}
    \label{fig:MANET}
\end{figure}

\textbf{Network Model and Mobility Model:} We consider a time slotted torus MANET of unit area. Similar to previous works, we assume that the network area is evenly partitioned into $m \times m$ cells as shown in Fig.~\ref{fig:network_partition}  \cite{Gamal_IT06,Sharma_TON07,Ciullo_TON11,Lipan_TMC12}. There are $n$ mobile nodes in the network and they randomly move around following the Independent and Identically Distributed (IID) mobility model \cite{Liu_TON12,Neely_IT05,Liu_TWC12}. According to the IID mobility model, each node first moves into a randomly and uniformly selected cell at the beginning of a time slot and then stays in that cell during the whole time slot. 

\textbf{Communication Model:} We assume that all nodes transmit data through one common wireless channel, and each node (say $S$ in Fig.~\ref{fig:network_partition}) employs the same transmission range $r=\sqrt{8}/m$ to cover $9$ cells, including $S$'s current cell and its $8$ neighboring cells. To account for mutual interference and interruption among concurrent transmissions, the commonly used protocol model is adopted \cite{Gupta_IT00,Kulkarni_IT04,Liu_TON12,Ciullo_TON11}. According to the protocol model, node $i$ could successfully transmit to another node $j$ if and only if $d_{ij} \leq r$ and for another simultaneously transmitting node $k \neq i,j$, $d_{kj} \geq (1+\Delta)\cdot r$, where $d_{ij}$ denotes the Euclidean distance between node $i$ and node $j$ and $\Delta \geq 0$ is the guard factor to prevent interference. In a time slot, the data  that can be  transmitted during a successful transmission is normalized to one packet.

\textbf{Traffic Model:} We consider the widely adopted permutation traffic model \cite{Liu_TON12,Ciullo_TON11,Liu_TWC12}, where there are $n$ distinct traffic flows in the network.
Under such traffic model, each node acts as the source of one traffic flow and at the same time the destination of another traffic flow.
The packet generating process in each source node is assumed to be a Bernoulli process, where a packet is generated by its source node with probability $\lambda$ in a time slot \cite{Neely_IT05}. 
We assume that each source node has a first-come-first-serve queue (called local-queue hereafter) with limited buffer size $M > 0$ to store its locally generated packets.
Each locally generated packet in a source node will be inserted into the end of its local-queue if the queue is not full, and dropped otherwise.

\subsection{MAC Protocol}
We adopt a commonly used MAC protocol based on the concept of Equivalent-Class to address wireless medium access issue in MANETs \cite{Liu_TON12,Ciullo_TON11,Lipan_TMC12,Kulkarni_IT04}. 
As illustrated in Fig.~\ref{fig:network_mac} that an Equivalent-Class (EC) is consisted of a group of cells with any two of them being separated by a horizontal and vertical distance of some integer multiple of $\alpha$ ($1\leq \alpha \leq m$) cells. 
 Under the EC-based MAC protocol (MAC-EC), the whole network cells are divided into $\alpha^2$ ECs and ECs are then activated alternatively from time slot to time slot. We call cells in an activated EC as active cells, and only a node in an active cell could access the wireless channel and do packet transmission. If there are multiple nodes in an active cell, one of them is selected randomly to have a fair access to wireless channel. 

To avoid interference among concurrent transmissions under the MAC-EC protocol, the parameter $\alpha$ should be set properly.
Suppose a node (say $S$ in Fig.~\ref{fig:network_mac}) in an active cell is transmitting to node $R$ at the current time slot, and another node $W$ in one adjacent active cell is also transmitting simultaneously. As required by the protocol model, the distance $d_{WR}$ between $W$ and $R$ should satisfy the following condition to guarantee successful transmission from $S$ to $R$, 
\begin{align} \label{eq:alpha_set1}
d_{WR} &\geq (1+\Delta)\cdot r
\end{align}
Notice that $d_{WR} \geq (\alpha-2)/m$, we have 
\begin{align} \label{eq:alpha_set2}
(\alpha-2)/m &\geq (1+\Delta)\cdot r	
\end{align} 
Since $\alpha \leq m$ and $r=\sqrt{8}/m$, $\alpha$ should be set as 
\begin{align}
\alpha &= \min\{\lceil (1+\Delta)\sqrt{8}+2\rceil,m\},
\end{align}
where the function $\lceil x \rceil$ returns the least integer value greater than or equal to $x$.

\subsection{PD-$f$ Scheme}
Once a node (say $S$) got access to the wireless channel in a time slot, it then executes the PD-$f$ scheme summarized in Algorithm~\ref{algorithm:PD-f} for packets dispatch.

\begin{remark}
The  PD-$f$ scheme is general and covers many widely used packet dispatching schemes as special cases,
like the ones without packet redundancy \cite{Grossglauser_TON02,Neely_IT05,Gamal_IT06} when $f=1$ and  only unicast transmission is allowed, the ones with controllable packet redundancy \cite{Small_WDTN05,Liu_TWC11,Liu_TON12} when $f>1$ and only unicast transmission is allowed, and the ones with uncontrollable packet redundancy \cite{Williams_MobiHoc02,Gao_WiOpt13} when $f \geq 1$ and broadcast transmission is allowed. 

\end{remark}

\floatname{algorithm}{Algorithm}

\begin{algorithm}[!ht]
\caption{PD-$f$ scheme}
\label{algorithm:PD-f}
\begin{algorithmic}[1]
	\IF{$S$ has packets in its local-queue}
		\STATE $S$ checks whether its destination $D$ is within its transmission range;
	
		\IF{$D$ is within its transmission range}
			\STATE $S$ transmits the head-of-line (HoL) packet in its local-queue to $D$; \{source-destination transmission\}
			\STATE $S$ removes the HoL packet from its local-queue;
			\STATE $S$ moves ahead the remaining packets in its local-queue;
		\ELSE
			\STATE With probability $q$ ($0 < q < 1$), $S$ dispatches the HoL packet;
			\IF{$S$ conducts packet dispatch}
				\STATE $S$ dispatches the HoL packet for one time; \{packet-dispatch transmission\}
				\IF{$S$ has already dispatched the HoL packet for $f$ times}
					\STATE $S$ removes the HoL packet from its local-queue;
					\STATE $S$ moves ahead the remaining packets in its local-queue;
				\ENDIF
			\ENDIF
		\ENDIF
	\ELSE
		\STATE $S$ remains idle;
	\ENDIF
\end{algorithmic}
\end{algorithm}


\section{QBD-Based Theoretical Framework} \label{section:framework}

In this section, a QBD-based theoretical framework is developed to capture the packet dispatching process in a PD-$f$ MANET. This framework will help us to analyze source delay in Section~\ref{section:delay}.

\subsection{QBD Modeling}

Due to the symmetry of source nodes, we only focus on a source node $S$ in our analysis. We adopt a two-tuple $\mathbf{X}(t)=(L(t), J(t))$ to define the state of the local-queue in $S$ at time slot $t$, where $L(t)$ denotes the number of packets in the local-queue at slot $t$ and $J(t)$ denotes the number of packet dispatches that have been conducted for the current head-of-line packet by slot $t$, here $0 \leq L(t) \leq M$, $0 \leq J(t) \leq f-1$ when $1 \leq L(t) \leq M$, and $J(t)=0$ when $L(t)=0$. 

Suppose that the local-queue in $S$ is at state $(l, j)$ in the current time slot, all the possible state transitions that may happen at the next time slot are summarized in Fig.~\ref{fig:state_transition}, where 
\begin{itemize}

\item $I_{0}(t)$ is an indicator function, taking value of $1$ if $S$ conducts source-destination transmission in the current time slot, and taking value of $0$ otherwise;

\item $I_{1}(t)$ is an indicator function, taking value of $1$ if $S$ conducts packet-dispatch transmission in the current time slot, and taking value of $0$ otherwise;

\item $I_{2}(t)$ is an indicator function, taking value of $1$ if $S$ conducts neither source-destination nor packet-dispatch transmission in the current time slot, and taking value of $0$ otherwise;

\item $I_{3}(t)$ is indicator function, taking value of $1$ if $S$ locally generates a packet in the current time slot, and taking value of $0$ otherwise.

\end{itemize}

\begin{figure}[!t]
\centering
{
			\centerline
			{
			\subfloat[State transition when $l=0$.]{\includegraphics[width=1.7in]{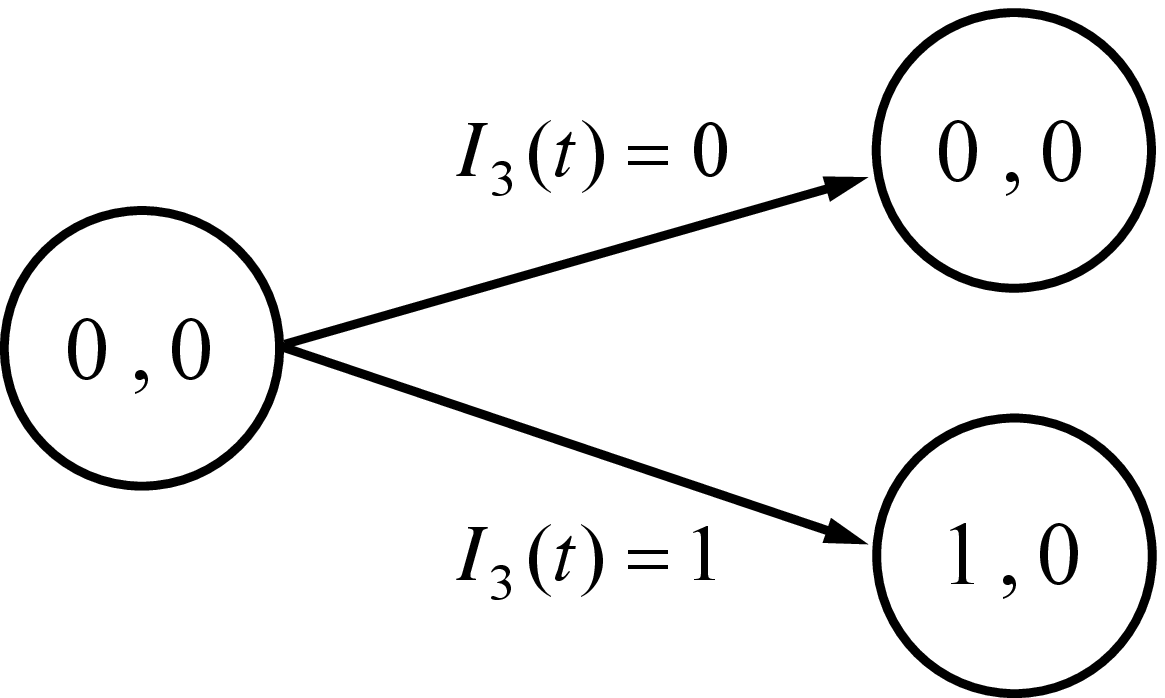} \label{fig:l0}}
			}
			\centerline
			{
			\subfloat[State transition when $1 \leq l \leq M-1$.]{\includegraphics[width=3.2in]{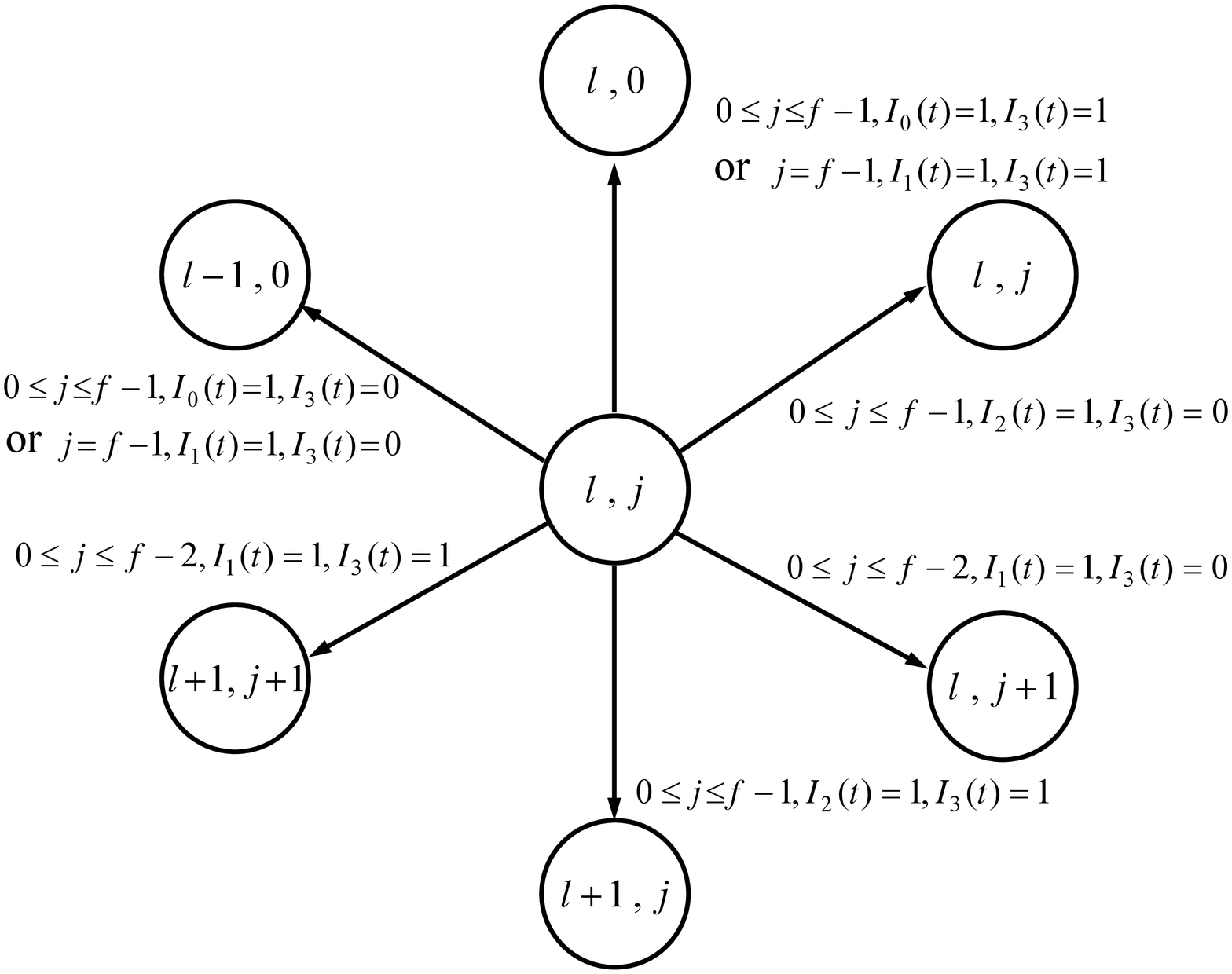} \label{fig:ll}}
			}
			\centerline
			{
			\subfloat[State transition when $l=M$.]{\includegraphics[width=3.2in]{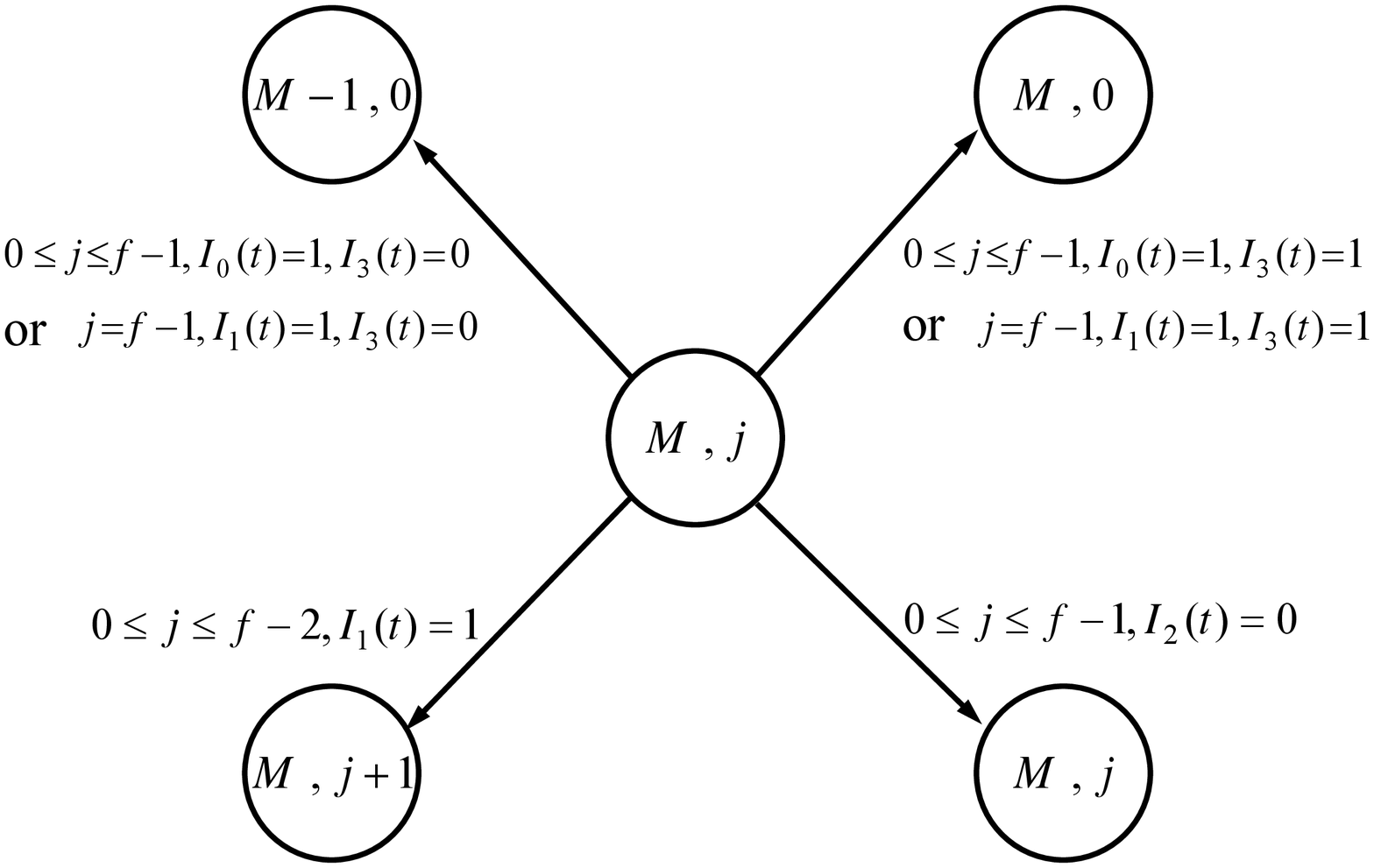} \label{fig:lM}}
			}
}
\caption{State transitions from state $(l, j)$ of the local-queue.}
\label{fig:state_transition}
\end{figure}

From Fig.~\ref{fig:state_transition} we can see that as time evolves, the state transitions of the local-queue in $S$ form a two-dimensional QBD process \cite{Latouche_Book99} 
\begin{align}	\label{eq:QBD}
\{\mathbf{X}(t), t = 0, 1, 2, \cdots \}, 
\end{align}
on state space
\begin{align}	\label{eq:QBD_space}
\big\{ \{(0,0)\} \cup \{(l,j)\}; 1 \leq l \leq M, 0 \leq j \leq f-1 \big\}.
\end{align}
Based on the transition scenarios in Fig.~\ref{fig:state_transition}, the overall transition diagram of above QBD process is illustrated in Fig.~\ref{fig:QBD}. 

\begin{figure*}[!t]
\centering
\includegraphics[width=4.5in]{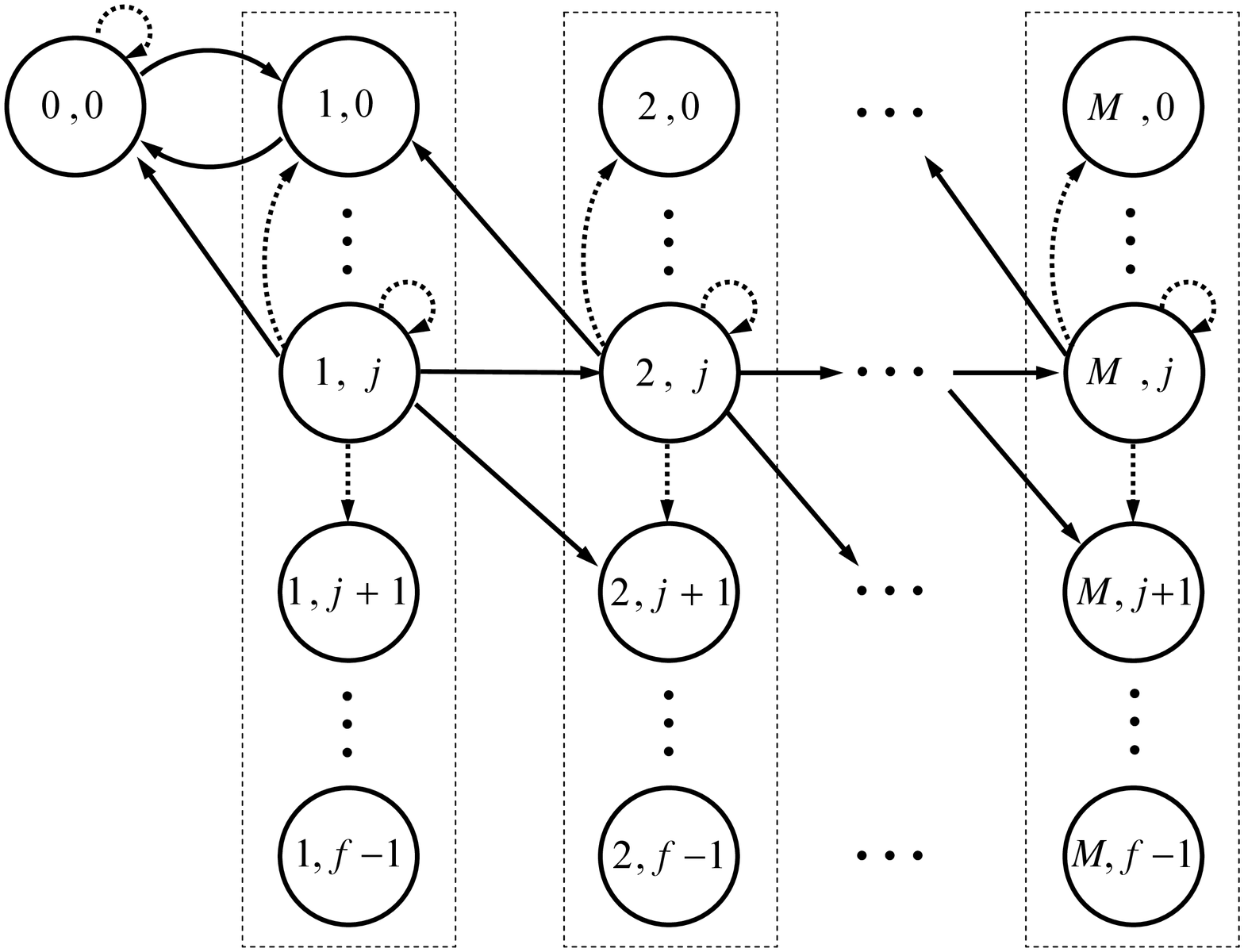}
\caption{State transition diagram for the QBD process of local-queue. For simplicity, only transitions from typical states $(l, j)$ are illustrated for $1 \leq l \leq M$, while other transitions are the same as that shown in Fig.~\ref{fig:state_transition}.}
\label{fig:QBD}
\end{figure*}

\begin{remark}
The QBD framework is powerful in the sense it enables main network dynamics to be captured, like the dynamics involved in the packet generating process and these involved in the source-destination and packet-dispatch transmissions (i.e., node mobility, medium contention, interference and packet transmitting). 
\end {remark}

\subsection{Transition Matrix and Some Basic Results}
As shown in Fig.~\ref{fig:QBD} that there are in total $1+M\cdot f$ two-tuple states for the local-queue in $S$. To construct the transition matrix of the QBD process, we arrange all these $1+M\cdot f$ states in a left-to-right and top-to-down way as follows: $\{(0, 0), (1, 0), (1, 1), \cdots, (1, f-1), (2, 0), (2, 1), \cdots, (2, f-1),\cdots, (M,0), \cdots, (M, f-1)\}$. 
Under such state arrangement, the corresponding state transition matrix $\mathbf{P}$ of the QBD process can be determined as 
\begin{align}	\label{eq:QBD_matrix}
\bf{P}=& 	
\left[
					\begin{array}{ccccc}
									\mathbf{B_1} &	\mathbf{B_0}	&		&		&	 \\ 
									\bf{B_2} &	\bf{A_1}	&	\bf{A_0}	&		&	 \\
										&	\bf{A_2}	&	\ddots	&	\ddots	&	 \\
										&						&	\ddots	&	\bf{A_1}	&	\bf{A_0} \\
									 	&						&					&	\bf{A_2}	&	\bf{A_M} \\
					\end{array}
\right],  
\end{align}
where the corresponding sub-matrices in matrix $\bf{P}$ are defined as follows:
\begin{itemize}
\item $\mathbf{B_0}$: a matrix of size $1 \times f$, denoting the transition probabilities from $(0,0)$ to $(1,j)$, $0 \leq j \leq f-1$.
\item $\mathbf{B_1}$: a matrix of size $1 \times 1$, denoting the transition probability from $(0,0)$ to $(0,0)$.
\item $\mathbf{B_2}$: a matrix of size $f \times 1$, denoting the transition probabilities from $(1,j)$ to $(0,0)$, $0 \leq j \leq f-1$.
\item $\mathbf{A_0}$: a matrix of size $f \times f$, denoting the transition probabilities from $(l,j)$ to $(l+1,j')$, $1 \leq l \leq M-1, 0 \leq j, j' \leq f-1$.
\item $\mathbf{A_1}$: a matrix of size $f \times f$, denoting the transition probabilities from $(l,j)$ to $(l,j')$, $1 \leq l \leq M-1, 0 \leq j, j' \leq f-1$.
\item $\mathbf{A_2}$: a matrix of size $f \times f$, denoting the transition probabilities from $(l,j)$ to $(l-1,j')$, $2 \leq l \leq M, 0 \leq j, j' \leq f-1$.
\item $\mathbf{A_M}$: a matrix of size $f \times f$, denoting the transition probabilities from $(M,j)$ to $(M,j')$, $0 \leq j, j' \leq f-1$.
\end{itemize}

Some basic probabilities involved in the above sub-matrices are summarized in the following Lemma.
\begin{lemma} \label{lemma:basic_p}
For a given time slot, let $p_{0}$ be the probability that $S$ conducts a source-destination transmission, let $p_{1}$ be the probability that $S$ conducts a packet-dispatch transmission, and let $p_2$ be the probability that $S$ conducts neither source-destination nor packet-dispatch transmission. Then, we have
\begin{align}
p_{0} &=	\frac{1}{\alpha^2}\bigg\{ \frac{9n-m^2}{n(n-1)}-\bigg(\frac{m^2-1}{m^2}\bigg)^{n-1}\frac{8n+1-m^2}{n(n-1)}\bigg\},	\\
p_{1} &= \frac{q(m^2-9)}{\alpha^2(n-1)}\bigg\{1-\bigg(\frac{m^2-1}{m^2}\bigg)^{n-1}\bigg\}, \\
p_{2} &= 1-p_{0}-p_{1}.
\end{align}

\end{lemma}
\begin{proof}
The proof is given in Appendix~\ref{Appendix:basic_p}.
\end{proof}

\section{Source Delay Analysis} \label{section:delay}
Based on the QBD-based theoretical framework developed above, this section conducts analysis on the source delay defined as follow. 

\begin{definition}
In a PD-$f$ MANET, the source delay $U$ of a packet is defined as the time the packet experiences in its local-queue after it is inserted into the local-queue.
\end{definition}

To analyze the source delay, we first examine the steady state distribution of the local-queue, based on which we then derive the CDF and mean/variance of the source delay.


\subsection{State Distribution of Local-Queue}
We adopt a row vector $\boldsymbol{\pi}^{*}_{\omega}=[\pi^{*}_{\omega,0} \,\, \boldsymbol{\pi}^{*}_{\omega,1} \cdots \boldsymbol{\pi}^{*}_{\omega,M}]$ of size $1+M \cdot f$ to denote the steady state distribution of the local-queue, here  $\pi^{*}_{\omega,0}$ is a scalar value representing the probability that the local-queue is in the state $(0,0)$, while  $\boldsymbol{\pi}^{*}_{\omega,l}=({\pi}^{*}_{\omega,l,j})_{1 \times f}$  is a sub-vector with ${\pi}^{*}_{\omega,l,j}$ being the probability that the local queue is in state $(l,j)$, $1 \leq l \leq M, 0 \leq j \leq f-1$. 

For the analysis of source delay, we further define a row vector $\boldsymbol{\pi}^{*}_{\Omega}=[\pi^{*}_{\Omega,0} \,\, \boldsymbol{\pi}^{*}_{\Omega,1} \cdots \boldsymbol{\pi}^{*}_{\Omega,M}]$ of size $1+M \cdot f$ to denote the conditional steady state distribution of the local-queue under the condition that a new packet has just been inserted into the local-queue, here $\pi^{*}_{\Omega,0}$ is a scalar value representing the probability that the local-queue is in the state $(0,0)$ under the above condition, while $\boldsymbol{\pi}^{*}_{\Omega,l}=({\pi}^{*}_{\Omega,l,j})_{1 \times f}$  is a sub-vector with ${\pi}^{*}_{\Omega,l,j}$ being the probability that the local queue is in state $(l,j)$ under the above condition, $1 \leq l \leq M, 0 \leq j \leq f-1$. 
Regarding the evaluation of $\boldsymbol{\pi}^{*}_{\Omega}$, we have the following lemma.

\begin{lemma}	\label{lemma:pi0}
In a PD-$f$ MANET, its conditional steady local-queue state distribution $\boldsymbol{\pi}^{*}_{\Omega}$ is given by 
\begin{align}
\boldsymbol{\pi}^{*}_{\Omega} &= \frac{\boldsymbol{\pi}^{*}_{\omega}\mathbf{P_2}}{\lambda\boldsymbol{\pi}^{*}_{\omega} \mathbf{P_1} \mathbf{1}}	\label{eq:pi},
\end{align}
where $\mathbf{1}$ is a column vector with all elements being $1$.
The matrix $\mathbf{P_1}$ in (\ref{eq:pi}) is determined based on (\ref{eq:QBD_matrix}) by setting the corresponding sub-matrices as follows:

For $M=1$,
\begin{align}
\mathbf{B_0} &= \mathbf{0},  \label{eq:P11-B0} \\
\mathbf{B_1} &= [1],  \label{eq:P11-B1}\\
\mathbf{B_2} &= \mathbf{c},  \label{eq:P11-B2} \\
\mathbf{A_M} &= \mathbf{0}.  \label{eq:P11-AM}
\end{align}

For $M \geq 2$,
\begin{align}
\mathbf{B_0} &= \mathbf{0},  \label{eq:P1-B0}\\
\mathbf{B_1} &= [1],  \label{eq:P1-B1}\\
\mathbf{B_2} &= \mathbf{c},  \label{eq:P1-B2}\\
\mathbf{A_0} &= \mathbf{0},  \label{eq:P1-A0}\\
\mathbf{A_1} &= \mathbf{Q},  \label{eq:P1-A1}\\
\mathbf{A_2} &= \mathbf{c} \cdot \mathbf{r},  \label{eq:P1-A2}\\
\mathbf{A_M} &= \mathbf{0}.  \label{eq:P1-AM}
\end{align}
where $\mathbf{0}$ is a matrix of proper size with all elements being $0$,
\begin{align}
\mathbf{c} = & [p_{0} \quad \cdots \quad p_{0} \quad p_{0}+p_{1}]^T,
\label{eq:matix_c} \\
\mathbf{r} = & [1 \quad 0 \quad \cdots \quad 0],	\label{eq:matix_r} \\
\bf{Q}=			& 	
\left[
					\begin{array}{ccccc}
									p_2 &	p_{1}	&		&		&	 \\ 
									  &	p_2				&	p_{1}	&		&	 \\
										&						&	\ddots	&	\ddots	&	 \\
										&						&					&	p_2			&	p_{1} \\
									 	&						&					&					&	p_2 \\
					\end{array}
\right]. \label{eq:matix_Q}
\end{align}

The matrix $\mathbf{P_2}$ in (\ref{eq:pi}) is also determined based on (\ref{eq:QBD_matrix}) by setting the corresponding sub-matrices as follows:

For $M=1$,
\begin{align}
\mathbf{B_0} &= [\lambda \quad 0 \quad \cdots \quad 0],  \label{eq:P21-B0}\\
\mathbf{B_1} &= [0],  \label{eq:P21-B1}\\
\mathbf{B_2} &= \mathbf{0},  \label{eq:P21-B2}\\
\mathbf{A_M} &= \lambda \mathbf{c}\cdot\mathbf{r}.  \label{eq:P21-AM}
\end{align}

For $M \geq 2$,
\begin{align}
\mathbf{B_0} &= [\lambda \quad 0 \quad \cdots \quad 0],  \label{eq:P2-B0}\\
\mathbf{B_1} &= [0],  \label{eq:P2-B1}\\
\mathbf{B_2} &= \mathbf{0},  \label{eq:P2-B2}\\
\mathbf{A_0} &= \lambda \mathbf{Q},  \label{eq:P2-A0}\\
\mathbf{A_1} &= \lambda \mathbf{c}\cdot\mathbf{r},  \label{eq:P2-A1}\\
\mathbf{A_2} &= \mathbf{0},  \label{eq:P2-A2}\\
\mathbf{A_M} &= \lambda \mathbf{c}\cdot\mathbf{r}.  \label{eq:P2-AM}
\end{align}
\end{lemma}

\begin{proof}
See Appendix~\ref{appendix:pi0} for the proof.
\end{proof}

The result in (\ref{eq:pi}) indicates that for the evaluation of $\boldsymbol{\pi}^{*}_{\Omega}$, we still need to determine the steady state distribution $\boldsymbol{\pi}^{*}_{\omega}$ of the local-queue.  

\begin{lemma} \label{lemma:pi_omega}
In a PD-$f$ MANET, its steady state distribution $\boldsymbol{\pi}^{*}_{\omega}$ of the local-queue is determined as follows:

For $M=1$,
\begin{align}
\pi^{*}_{\omega,0} &= \pi^{*}_{\omega,0}\mathbf{B_1} + \boldsymbol{\pi}^{*}_{\omega,1}\mathbf{B_2}, \label{eq:piM11} \\
\boldsymbol{\pi}^{*}_{\omega,1} &= \pi^{*}_{\omega,0}\mathbf{B_0} + \boldsymbol{\pi}^{*}_{\omega,1}\mathbf{A_M}, \label{eq:piM12} \\
\boldsymbol{\pi}^{*}_{\omega}\cdot \mathbf{1} &=1. \label{eq:piM13} 
\end{align}

For $M=2$,
\begin{align}
\pi^{*}_{\omega,0} &= \pi^{*}_{\omega,0}\mathbf{B_1} + \boldsymbol{\pi}^{*}_{\omega,1}\mathbf{B_2}, \label{eq:piM21} \\
\boldsymbol{\pi}^{*}_{\omega,1} &= \pi^{*}_{\omega,0}\mathbf{B_0} + \boldsymbol{\pi}^{*}_{\omega,1}\mathbf{A_1}+ \boldsymbol{\pi}^{*}_{\omega,2}\mathbf{A_2}, \label{eq:piM22} \\
\boldsymbol{\pi}^{*}_{\omega,2} &= \boldsymbol{\pi}^{*}_{\omega,1}\mathbf{A_0}+ \boldsymbol{\pi}^{*}_{\omega,2}\mathbf{A_M}, \label{eq:piM23} \\
\boldsymbol{\pi}^{*}_{\omega}\cdot \mathbf{1} &=1. \label{eq:piM24} 
\end{align}

For $M \geq 3$,
\begin{align}
[\pi^{*}_{\omega,0}, \boldsymbol{\pi}^{*}_{\omega,1}]	&=	[\pi^{*}_{\omega,0}, \boldsymbol{\pi}^{*}_{\omega,1}]
																						\left[
																									\begin{array}{cc}
																									\mathbf{B_1}	&	\mathbf{B_0}	\\
																									\mathbf{B_2}	&	\mathbf{A_1}+\mathbf{R}\mathbf{A_2}
																									\end{array}
																						\right],		\label{eq:pi'_1} \\
\boldsymbol{\pi}^{*}_{\omega,i} &= \boldsymbol{\pi}^{*}_{\omega,1}\mathbf{R}^{i-1}, \quad 2 \leq i \leq M-1, \label{eq:pi'_2} \\
\boldsymbol{\pi}^{*}_{\omega,M} &= \boldsymbol{\pi}^{*}_{\omega,1}\mathbf{R}^{M-2}\mathbf{R_M},	\label{eq:pi'_3} \\
\boldsymbol{\pi}^{*}_{\omega}\cdot \mathbf{1} &=1 \label{eq:pi'_4}, 
\end{align}
where 
\begin{align}
\mathbf{B_0} &= [\lambda \quad 0 \quad \cdots \quad 0],  \label{eq:P3-B0} \\
\mathbf{B_1} &= [1-\lambda],  \label{eq:P3-B1} \\
\mathbf{B_2} &= (1-\lambda)\mathbf{c},  \label{eq:P3-B2} \\
\mathbf{A_0} &= \lambda \mathbf{Q},  \label{eq:P3-A0} \\
\mathbf{A_1} &= (1-\lambda)\mathbf{Q}+ \lambda \mathbf{c}\cdot\mathbf{r},  \label{eq:P3-A1} \\
\mathbf{A_2} &= (1-\lambda) \mathbf{c}\cdot\mathbf{r},  \label{eq:P3-A2} \\
\mathbf{A_M} &= \mathbf{A_1}+\mathbf{A_0},	 \label{eq:P3-AM} \\
\mathbf{R} &=\mathbf{A_0}[\mathbf{I}-\mathbf{A_1}-\mathbf{A_0}\cdot\mathbf{1}\cdot \mathbf{r}]^{-1}, \\
\mathbf{R_M} &=\mathbf{A_0}[\mathbf{I}-\mathbf{A_M}]^{-1},
\end{align}
here $\mathbf{c}$, $\mathbf{r}$ and $\mathbf{Q}$ are given in (\ref{eq:matix_c}), (\ref{eq:matix_r}) and (\ref{eq:matix_Q}), respectively; $\mathbf{I}$ is an identity matrix of size $f \times f$, and $\mathbf{1}$ is a column vector of proper size with all elements being $1$.
\end{lemma}

\begin{proof}
See Appendix~\ref{appendix:pi_omega} for the proof.
\end{proof}

\subsection{CDF, Mean and Variance of Source Delay}	\label{subsection:source_delay}
Based on the conditional steady state distribution $\boldsymbol{\pi}^{*}_{\Omega}$ of the local-queue, we are now ready to derive the CDF as well as mean and variance of the source delay, as summarized in the following theorem.

\begin{theorem}	\label{theorem:pi'}
In a PD-$f$ MANET, the probability mass function $Pr\{U=u\}$, CDF $Pr\{U\leq u\}$, mean $\overline{U}$ and variance ${\sigma}^2_U$ of the source delay $U$ of a packet are given by
\begin{align}
Pr\{U=u\}	&=\boldsymbol{\pi}^{-}_{\Omega}\mathbf{T}^{u-1}\mathbf{c}^{+}, \quad u \geq 1, \label{eq:PMF} \\
Pr\{U\leq u\} &= 1 - \boldsymbol{\pi}^{-}_{\Omega}\mathbf{T}^{u}\mathbf{1}, \quad u \geq 0,\label{eq:CDF} \\
\overline{U} &= \boldsymbol{\pi}^{-}_{\Omega}(\mathbf{I}-\mathbf{T})^{-2}\mathbf{c}^{+}	\label{lemma:mean_1}, \\
{\sigma}^2_U &= \boldsymbol{\pi}^{-}_{\Omega}(\mathbf{I}+\mathbf{T})(\mathbf{I}-\mathbf{T})^{-3}\mathbf{c}^{+}-{\overline{U}}^2,	\label{lemma:variance}
\end{align}
where $\boldsymbol{\pi}^{-}_{\Omega}=[\boldsymbol{\pi}^{*}_{\Omega,1} \,\, \boldsymbol{\pi}^{*}_{\Omega,2} \cdots \boldsymbol{\pi}^{*}_{\Omega,M}]$ is a sub vector of $\boldsymbol{\pi}^{*}_{\Omega}$, $\mathbf{c}^{+}$ is a column vector of size $M \cdot f$ and $\mathbf{T}$ is a matrix of size $(M \cdot f) \times (M \cdot f)$ determined as follows:

For $M=1$, 
\begin{align}
\mathbf{c}^{+} &= \mathbf{c}, \label{eq:cplus_M1} \\
\mathbf{T} &= \mathbf{Q}. \label{eq:T_M1}
\end{align}

For $M \geq 2$,
\begin{align}
\mathbf{c}^{+} &= [\mathbf{c} \quad 0 \quad \cdots \quad 0]^{T}, \label{eq:cplus_M2} \\
\mathbf{T} &= 
\left[
					\begin{array}{ccccc}
									\mathbf{A_1} &	\mathbf{A_0}	&		&		&	 \\ 
									\bf{A_2} &	\bf{A_1}	&	\bf{A_0}	&		&	 \\
										&	\ddots	&	\ddots	&	\ddots	&	 \\
										&						&	\bf{A_2}	&	\bf{A_1}	&	\bf{A_0} \\
									 	&						&					&	\bf{A_2}	&	\bf{A_M} \\
					\end{array}
\right],  \label{eq:T_M2}
\end{align}
where 

\begin{align}
\mathbf{A_0} &= \mathbf{0},  \\
\mathbf{A_1} &= \mathbf{Q},  \\
\mathbf{A_2} &= \mathbf{c}\cdot\mathbf{r}, \\
\mathbf{A_M} &= \mathbf{Q},
\end{align}
here $\mathbf{c}$, $\mathbf{r}$ and $\mathbf{Q}$ are given in (\ref{eq:matix_c}), (\ref{eq:matix_r}) and (\ref{eq:matix_Q}), respectively, and $\mathbf{0}$ is a matrix of proper size with all elements being $0$.

%
%

\end{theorem}

\begin{proof}
See Appendix~\ref{appendix:theorem:pi'} for the proof.
\end{proof}

%

\section{Numerical Results} \label{section:numerical}
In this section, we first provide simulation results to validate the efficiency of our QBD-based theoretical framework and source delay models, and then illustrate how source delay in a PD-$f$ MANET is related to network parameters. 


\subsection{Source Delay Validation}

To validate the theoretical framework and source delay models, a customized C++ simulator was developed to simulate the packet generating and dispatching processes in PD-$f$ MANETs \cite{SourceDelay}, in which network parameters, such as the number of network nodes $n$, network partition parameter $m$, local-queue buffer size $M$, packet dispatch limit $f$, packet dispatch probability $q$ and packet generating probability $\lambda$, can be flexibly adjusted to simulate source delay performance under various network scenarios. 
Based on the simulator, extensive simulations have been conducted to validate our our QBD-based source delay models. For three typical network scenarios of $n=100$ (small network), $n=200$ (medium network) and $n=400$ (large network) with $m = 8, M = 7, f = 2, q = 0.4$ and $\lambda = 0.001$, the corresponding simulation/theoretical results on the CDFs of source delay are summarized in Fig.~\ref{fig:CDF}.

We can see from Fig.~\ref{fig:CDF} that for all three network scenarios considered here, the theoretical results on the CDF of source delay match nicely with the corresponding simulated ones, indicating that our QBD-based theoretical framework is highly efficient in modeling the source delay behaviors of PD-$f$ MANETs. 
We can also see from Fig.~\ref{fig:CDF} that the source delay in a small network (e.g. $n=100$ here) is very likely smaller than that of a large network (e.g. $n=200$ or $n=400$ here). This is because that for a given network area and a fixed partition parameter $m$, as network size in terms of $n$ decreases the channel contention becomes less severe and thus each source node has more chances to conduct packet dispatch, leading to a shorter source delay one packet experiences in its source node.



\begin{figure}[!t]
\centering
\includegraphics[width=3.5in]{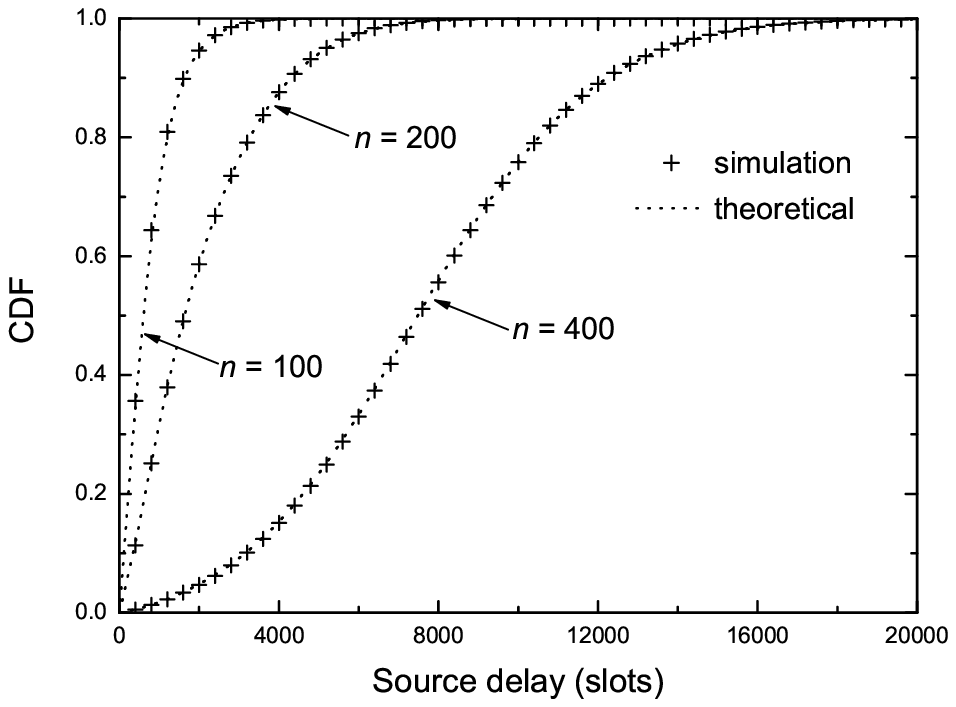}
\caption{The simulation and theoretical results on cumulative distribution function (CDF) of source delay.}
\label{fig:CDF}
\end{figure}

\subsection{Source Delay Illustrations}

With our QBD-based theoretical framework, we then illustrate how source delay performance, in terms of its mean $\overline{U}$ and standard deviation $\sigma_{U}=\sqrt{\sigma_{U}^{2}}$, is related to some main network parameters like packet generating probability $\lambda$, local-queue buffer size $M$, packet dispatch limit $f$ and packet dispatch probability $q$.

We first illustrate in Figs.~\ref{fig:meanvarbuffer} how $\overline{U}$ and $\sigma_{U}$ vary with $\lambda$ and $M$ for a network scenario of $n=200, m=16, q=0.6$ and $f=3$. 
We see from Fig.~\ref{fig:MeanxMyL} that for any given $M$, $\overline{U}$ first increases as $\lambda$ increases until $\lambda$ reaches some threshold value and then $\overline{U}$ remains almost a constant as $\lambda$  increases further beyond that threshold. On the other hand, for a given $\lambda \in[0.0005, 0.002]$, as $M$ increases $\overline{U}$ first increases and then remains constant, while for a given $\lambda \in[0.002, 0.01]$, $\overline{U}$ always increases as $M$ increases. 
Regarding the standard deviation $\sigma_{U}$ of source delay, we see from Fig.~\ref{fig:VarxMyL} that for given $M$, as $\lambda$ increases from 0.0005 to 0.01 $\sigma_{U}$ first increases sharply to a peak value, then decreases sharply, and finally converges to a  constant. It is interesting to see that the peak values of $\sigma_{U}$ under different settings of $M$ are all achieved at the same $\lambda=0.0025$. The results in Fig.~\ref{fig:VarxMyL} further indicate that for fixed $\lambda$, as $M$ increases $\sigma_{U}$ always first increases and then gradually converges to a constant.

\begin{figure}[!t]
    \centering
    {
    \subfloat[$\overline{U}$ versus $(\lambda, M)$]
    {\includegraphics[width=4in]{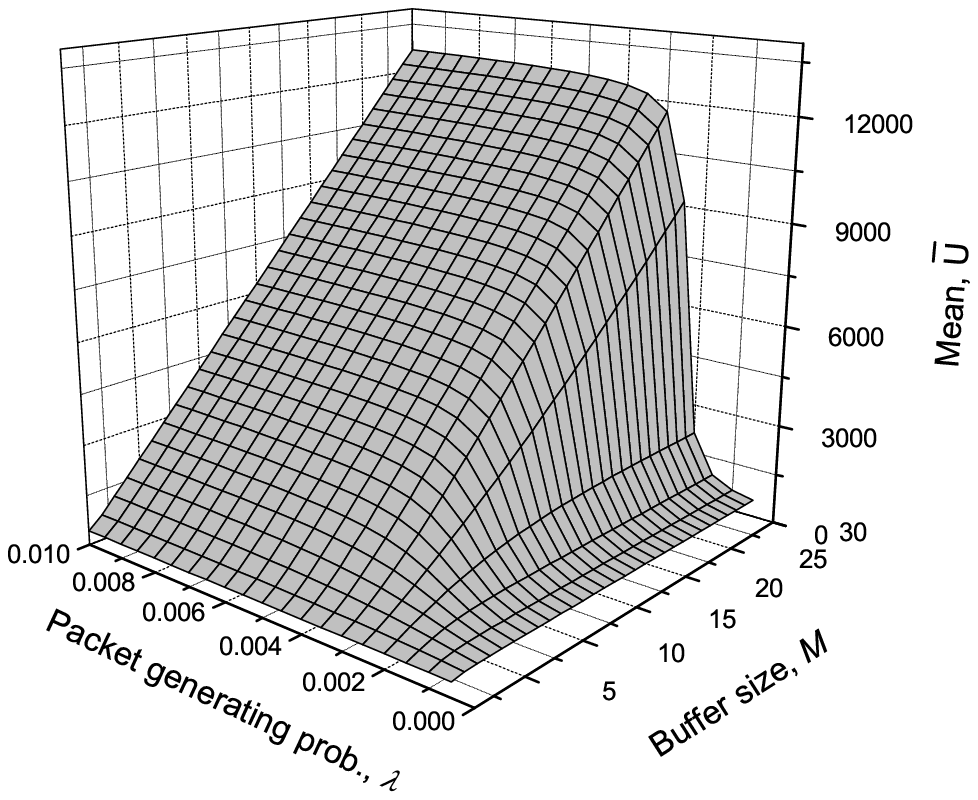} \label{fig:MeanxMyL}}
   	\hfill
    \subfloat[${\sigma}_U$ versus $(\lambda, M)$]
    {\includegraphics[width=4in]{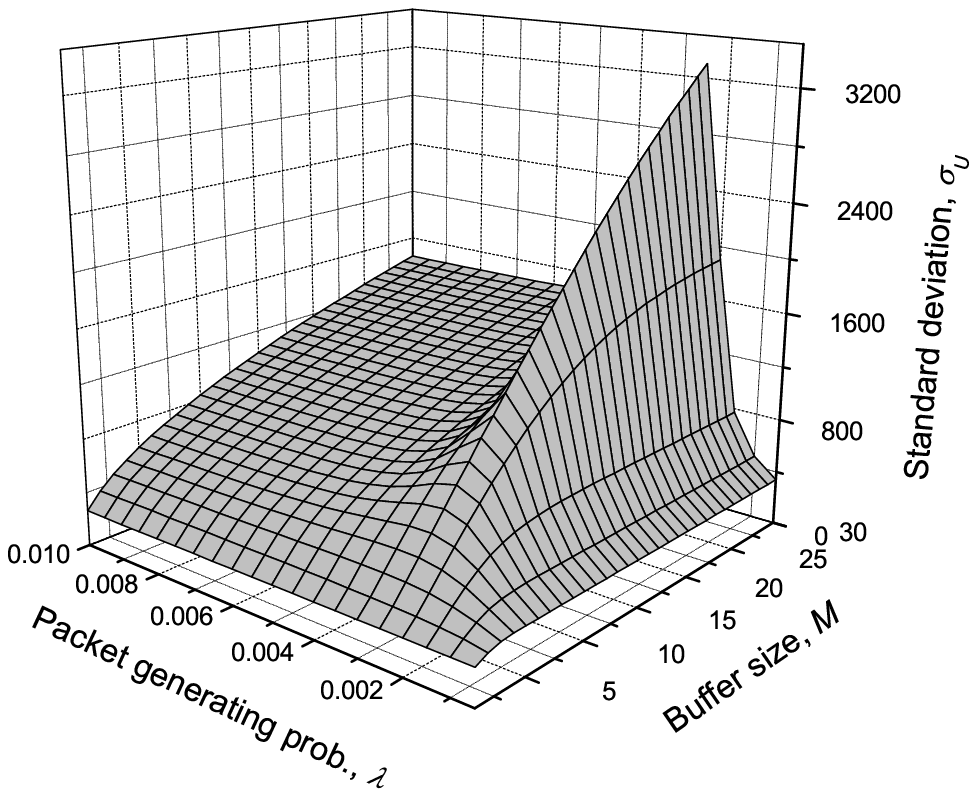} \label{fig:VarxMyL}}
    }
    \caption{Source delay performance versus  packet generating probability $\lambda$ and local-queue buffer size $M$.}
    \label{fig:meanvarbuffer}
\end{figure}

We then illustrate in Figs.~\ref{fig:meanvar} how $\overline{U}$ and $\sigma_{U}$ vary with packet dispatch parameters $q$ and $f$ under the network scenario of $n=300, m=16, M=7$ and $\lambda=0.002$. 
From Fig.~\ref{fig:mean} and Fig.~\ref{fig:var} we can see that although both $\overline{U}$ and $\sigma_{U}$ always decrease as $q$ increases for a fixed $f$, their variations with $q$ change dramatically with the setting of $f$. 
On the other hand, for a given $q \in[0.05, 0.2]$,  as $f$ increases both $\overline{U}$ and $\sigma_{U}$ first increase and then tend to a constant, while for a given $q \in [0.2,,1.0]$, both $\overline{U}$ and $\sigma_{U}$ always monotonically increase as $f$ increases.


\begin{figure}[!t]
    \centering
    {
    \subfloat[$\overline{U}$ versus $(q, f)$]
    {\includegraphics[width=3.5in]{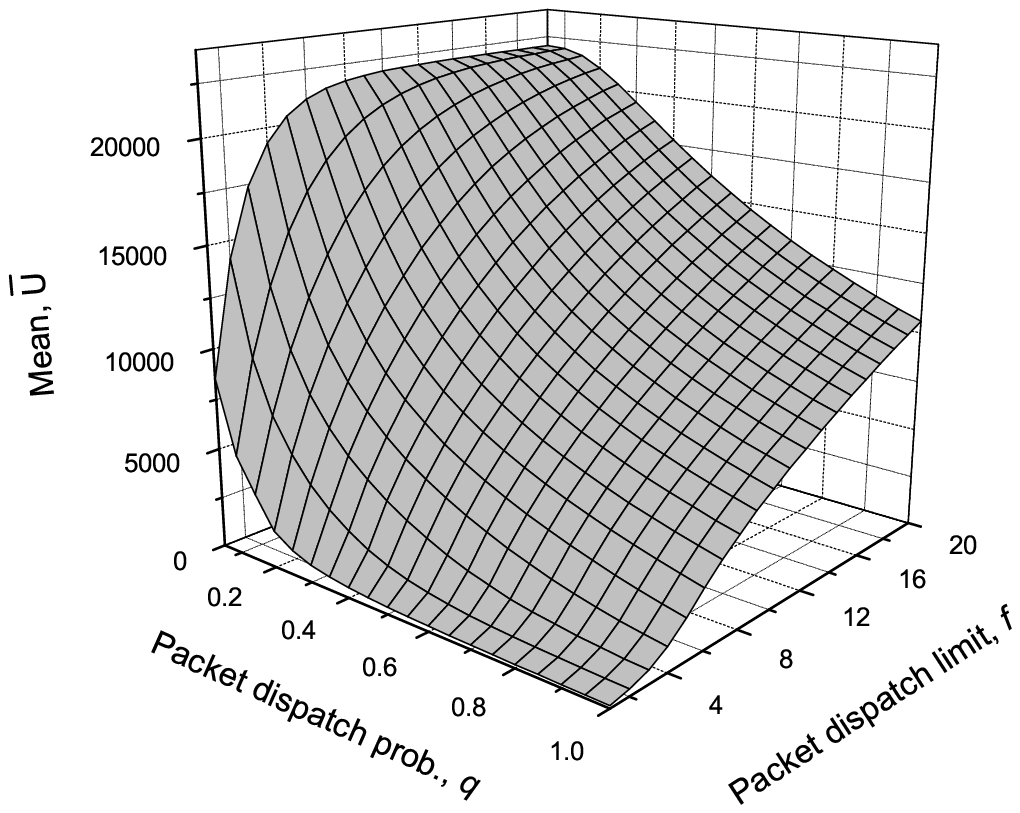} \label{fig:mean}}
   	\hfill
    \subfloat[$\sigma_{U}$ versus $(q, f)$]
    {\includegraphics[width=3.5in]{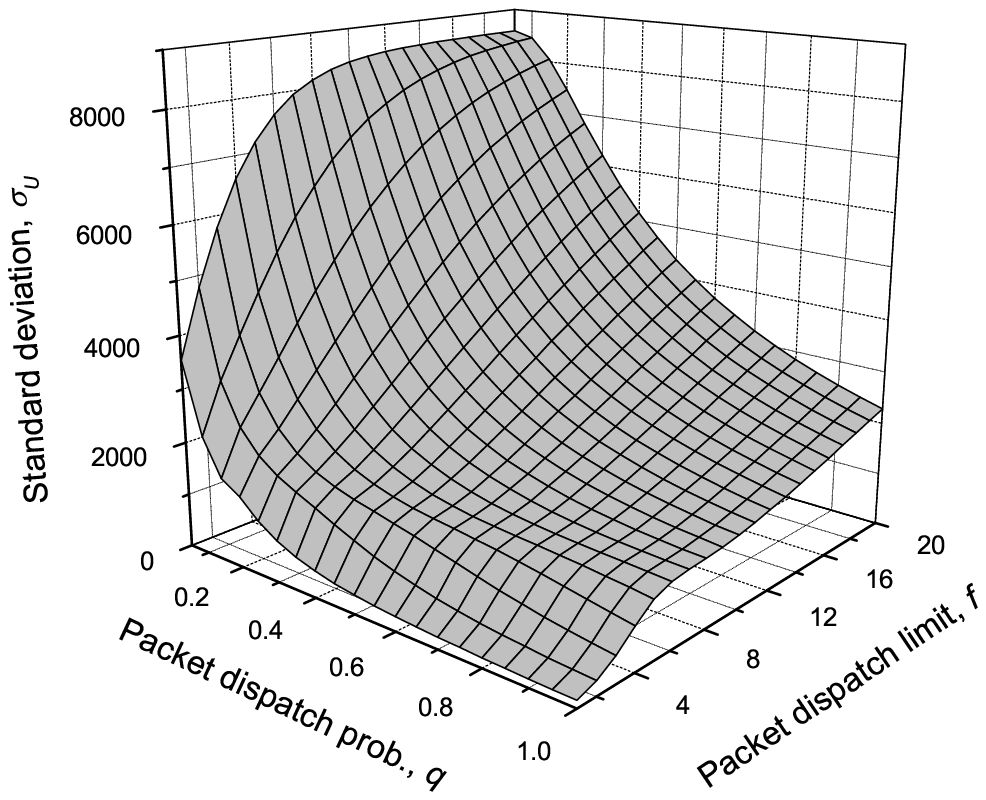} \label{fig:var} }
    }
    \caption{Source delay performance versus packet dispatch probability $q$ and packet dispatch limit $f$.}
    \label{fig:meanvar}
\end{figure}

\section{Related Works} \label{section:related_works}
A substantial amount of works have been devoted to the study of delay performance in MANETs, which can be roughly divided into partial delay study and overall delay study.

\subsection{Partial Delay Study}
The available works on partial delay study in MANETs mainly focus on the delivery delay analysis \cite{Hanbali_PE08,Small_WDTN05,Panagakis_WiOpt07,Ibrahim_PE07,Groenevelt_PE05,Hanbali_IFIP07,Liu_TON12,Spyropoulos_TON08} and local delay analysis \cite{Baccelli_INFOCOM10,Haenggi_IT13,Gong_TWC13}, which constitutes only a part of the overall packet delay. 

The delivery delay, defined as the time it takes a packet to reach its destination after its source starts to deliver it, has been extensively studied in the literature. 
For sparse MANETs without channel contentions, the Laplace-Stieltjes transform of delivery delay was studied in \cite{Groenevelt_PE05}; later, by imposing lifetime constraints on packets, the cumulative distribution function and $n$-th order moment of delivery delay were examined in \cite{Hanbali_PE08,Hanbali_IFIP07}; the delivery delay was also studied in \cite{Panagakis_WiOpt07,Small_WDTN05,Spyropoulos_TON08} under different assumptions on inter-meeting time among mobile nodes. For more general MANETs with channel contentions, closed-form results on mean and variance of delivery delay were recently reported in \cite{Liu_TON12}.
Regarding the local delay, i.e.  the time it takes a node to successfully transmit a packet to its next-hop receiver, it was reported in \cite{Baccelli_INFOCOM10} that some MANETs may suffer from a large and even infinite local delay. The work \cite{Haenggi_IT13} indicates that the power control serves as a very efficient approach to ensuring a finite local delay in MANETs. It was further reported in \cite{Gong_TWC13} that by properly exploiting node mobility in MANETs it is possible for us to reduce local delay there.


\subsection{Overall Delay Analysis}
Overall delay (also called end-to-end delay), defined as the time it takes a packet to reach its destination after it is generated at its source, has also been extensively studied in the literature. 
For MANETs with two-hop relay routing, closed-form upper bounds on expected overall delay were derived in \cite{Neely_IT05,Liu_TWC11}. 
For MANETs with two-hop relay routing and its variants, approximation results on expected overall delay were presented in \cite{Liu_WCNC12,Liu_TWC12}.
For MANETs with multi-hop relay routing, upper bounds on the cumulative distribution function of overall delay were reported in \cite{Ciucu_Allerton10,Ciucu_SIGMETRICS11}, and approximations on the expected overall delay were derived in \cite{Jindal_TMC09}. Rather than studying upper bounds and approximations on overall delay, some recent works explored the exact overall delay and showed that it is possible to derive the exact results on overall delay for MANETs under some special two-hop relay routings \cite{Neely_IT05,Gao_WiOpt13}.

%

\section{Conclusion}	\label{section:conclusion}
This paper conducted a thorough study on the source delay in MANETs, a new and fundamental delay metric for such networks. A QBD-based theoretical framework was developed to model the source delay behaviors under a general packet dispatching scheme, based on which the cumulative distribution function as well as the mean and variance of source delay were derived. As validated through extensive simulation results that our QBD-based framework is highly efficient in modeling the source delay performance in MANETs. Numerical results were also provided to illustrate how source delay is related with and thus can be controlled by some key network parameters, like local-queue buffer size, packet dispatch limit, and packet dispatch probability. It is expected that our source delay analysis and the related QBD-based theoretical framework will solidly contribute to the study of overall delay behavior in MANETs.


%

\appendices
\section{Proof of Lemma~\ref{lemma:basic_p}}	\label{Appendix:basic_p}
The proof process is similar to that in \cite{Liu_TWC11,Liu_TWC12}. We omit the proof details here and just outline the main idea of the proof. To derive the probability $p_{0}$ (resp. $p_{1}$), we first divide the event that $S$ conducts a  source-destination transmission (resp. packet-dispatch transmission) in a time slot into following sub-events: 1) $S$ moves into an active cell in the time slot according to the IID mobility model; 2) $S$ successfully accesses the wireless channel after fair contention according to the MAC-EC protocol; 3) $S$ selects to conduct source-destination transmission (resp. packet-dispatch transmission) according to the PD-$f$ scheme. We can then derive probability $p_{0}$ (resp. $p_{1}$) by combining the probabilities of these sub-events.

\section{Proof of Lemma~\ref{lemma:pi0}}	\label{appendix:pi0}

To derive the conditional steady state distribution $\boldsymbol{\pi}_{\Omega}^{*}$ of the local-queue under the condition that a packet has just been inserted into the queue, we first study its corresponding transient state distribution $\boldsymbol{\pi}_{\Omega}(t+1)$ at time slot $t+1$. 

Similar to the definition of $\boldsymbol{\pi}_{\Omega}^{*}$, we can see that the $(2+(l-1)f+j)$-th entry of row vector $\boldsymbol{\pi}_{\Omega}(t+1)$, denoted by $[\boldsymbol{\pi}_{\Omega}(t+1)]_{2+(l-1)f+j}$ here, corresponds to the probability that the local-queue is in state $\mathbf{X}(t+1)=(l, j)$ in time slot $t+1$ under the condition that a packet has just been inserted into the local-queue in time slot $t$, $1 \leq l \leq M, 0 \leq j \leq f-1$. The basic state transition from $\mathbf{X}(t)$ to $\mathbf{X}(t+1)$ is illustrated in  Fig.~\ref{fig:time_slot}, where $I_0(t)$ through $I_3(t)$ are indicator functions defined in Section III.A, and $I_4(t)$ is a new indicator function, taking value of $1$ if the local-queue is not full in time slot $t$ (i.e. the local-queue is in some state in $\big\{\{(0, 0)\}\cup \{(l ,j)\}; 1 \leq l \leq M-1, 0 \leq j \leq f-1\big\}$), and taking value of $0$ otherwise.

\begin{figure}[!t]
\centering
\includegraphics[width=3in]{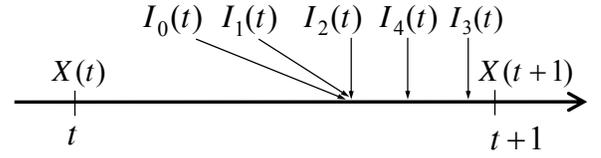}
\caption{Illustration for state transition from $\mathbf{X}(t)$ to $\mathbf{X}(t+1)$ during time slot $[t, t+1)$.}
\label{fig:time_slot}
\end{figure}

From Fig.~\ref{fig:time_slot} we can see that $[\boldsymbol{\pi}_{\Omega}(t+1)]_{2+(l-1)f+j}$ is evaluated as
\begin{align}
&[\boldsymbol{\pi}_{\Omega}(t+1)]_{2+(l-1)f+j}	\\
&=Pr\{\mathbf{X}(t+1)=(l, j) | I_4(t)=1, I_3(t)=1\}  \\
 &=\frac{Pr\{I_4(t)=1, I_3(t)=1, \mathbf{X}(t+1)=(l, j)\}}{Pr\{I_4(t)=1, I_3(t)=1\}} \\
&=\frac{Pr\{I_4(t)=1, I_3(t)=1, \mathbf{X}(t+1)=(l, j)\}}{\lambda \cdot Pr\{I_4(t)=1\}} \label{eq:equilibrium_2},
\end{align}
where (\ref{eq:equilibrium_2}) follows because the packet generating process is a Bernoulli process independent of the state of the local-queue.

For the probability $Pr\{I_4(t)=1\}$ in (\ref{eq:equilibrium_2}), we have 
\begin{align}
& Pr\{I_4(t)=1\} \\
&= \sum_{(l', j')}Pr\{I_4(t)=1, \mathbf{X}(t)=(l', j')\} \\
&= \sum_{(l', j')}Pr\{\mathbf{X}(t)=(l', j')\}Pr\{I_4(t)=1|\mathbf{X}(t)=(l', j')\}  	\label{eq:I(t2)}
\end{align}
where $Pr\{I_4(t)=1|\mathbf{X}(t)=(l', j')\}$ is actually the transition probability from state $\mathbf{X}(t)=(l', j')$ to states in $\big\{\{(0, 0)\}\cup \{(l ,j)\}; 1 \leq l \leq M-1, 0 \leq j \leq f-1\big\}$. The matrix $\mathbf{P_1}$ of such transition probabilities can be determined based on (\ref{eq:QBD_matrix}) by setting the corresponding sub-matrices according to (\ref{eq:P11-B0})-(\ref{eq:P1-AM}). With matrix $\mathbf{P_1}$ and (\ref{eq:I(t2)}), we have
\begin{align}
Pr\{I_4(t)=1\}	&= \boldsymbol{\pi}_{\omega}(t) \cdot \mathbf{P_1} \cdot \mathbf{1},	\label{eq:equilibrium_denominator}
\end{align}
where $\boldsymbol{\pi}_{\omega}(t)=(\boldsymbol{\pi}_{\omega,l,j}(t))_{1 \times M \cdot f}$ with $\boldsymbol{\pi}_{\omega,l,j}(t)$ being the probability $Pr\{\mathbf{X}(t)=(l', j')\}$.

For the numerator of (\ref{eq:equilibrium_2}), we have
\begin{align}
&Pr\{I_4(t)=1, I_3(t)=1, \mathbf{X}(t+1)=(l, j)\} \\
&=\!\!\!\!\sum_{\substack{(l', j')}}\!\!\!Pr\{\mathbf{X}(t)\!=\!(l'\!, j'), I_4(t)\!=\!1, I_3(t)\!=\!1, \mathbf{X}(t\!+\!1)\!=\!(l, j)\} \\
&=\sum_{\substack{(l', j')}}\!Pr\{\mathbf{X}(t)\!=\!(l', j')\}	\nonumber \\
&\quad \cdot Pr\{I_4(t)\!=\!1, I_3(t)\!=\!1, \mathbf{X}(t\!+\!1)\!=\!(l, j) | \mathbf{X}(t)\!=\!(l', j')\},  \label{eq:X(t+1)}
\end{align}
where $Pr\{I_4(t)\!=\!1, I_3(t)\!=\!1, \mathbf{X}(t\!+\!1)\!=\!(l, j)| \mathbf{X}(t)\!=\!(l', j')\}$ represents the transition probability from state $\mathbf{X}(t)\!=\!(l', j')$ to state $\mathbf{X}(t\!+\!1)\!=\!(l, j)$, with the condition that events $\{I_4(t)=1\}$ and $\{I_3(t)=1\}$ also happen simultaneously. The matrix $\mathbf{P_2}$ of such transition probabilities is determined based on (\ref{eq:QBD_matrix}) by setting the corresponding sub-matrices according to (\ref{eq:P21-B0})-(\ref{eq:P2-AM}).
With matrix $\mathbf{P_2}$ and (\ref{eq:X(t+1)}), we have
\begin{align}
& Pr\{I_4(t)=1, I_3(t)=1, \mathbf{X}(t+1)=(l, j)\} \nonumber \\
&= [\boldsymbol{\pi}_{\omega}(t)\mathbf{P_2}]_{2+(l-1)f+j}. \label{eq:equilibrium_numerator}
\end{align}

After substituting (\ref{eq:equilibrium_denominator}) and (\ref{eq:equilibrium_numerator}) into (\ref{eq:equilibrium_2}), we get
\begin{align}
 &[\boldsymbol{\pi}_{\Omega}(t+1)]_{2+(l-1)f+j} \nonumber \\
 &=\frac{[\boldsymbol{\pi}_{\omega}(t)\mathbf{P_2}]_{2+(l-1)f+j}}{\lambda\boldsymbol{\pi}_{\omega}(t) \mathbf{P_1} \mathbf{1}}.
\end{align}
Thus, in vector form
\begin{align}
\boldsymbol{\pi}_{\Omega}(t+1) =\frac{\boldsymbol{\pi}_{\omega}(t)\mathbf{P_2}}{\lambda\boldsymbol{\pi}_{\omega}(t) \mathbf{P_1} \mathbf{1}}. \label{eq:pi(t+1)}
\end{align}

Taking limits on both sides of (\ref{eq:pi(t+1)}), we get the steady state distribution $\boldsymbol{\pi}_{\Omega}^{*}$ as 
\begin{align}
\boldsymbol{\pi}_{\Omega}^{*} &= \lim_{t \rightarrow \infty}\boldsymbol{\pi}_{\Omega}(t+1) \\
	&=\lim_{t \rightarrow \infty} \frac{\boldsymbol{\pi}_{\omega}(t)\mathbf{P_2}}{\lambda\boldsymbol{\pi}_{\omega}(t) \mathbf{P_1} \mathbf{1}}	\label{eq:pi_limit} \\
	&= \frac{\boldsymbol{\pi}^{*}_{\omega}\mathbf{P_2}}{\lambda\boldsymbol{\pi}^{*}_{\omega} \mathbf{P_1} \mathbf{1}},
\end{align}
where 
\begin{align}
\boldsymbol{\pi}_{\omega}^{*} = \lim_{t \rightarrow \infty} \boldsymbol{\pi}_{\omega}(t). \label{eq:pi'define}
\end{align}

This completes the proof of Lemma~\ref{lemma:pi0}.

\section{Proof of Lemma~\ref{lemma:pi_omega}}	\label{appendix:pi_omega}


Recall that as time evolves, the state transitions of the local-queue form a QBD process shown in Fig.~\ref{fig:QBD}.
From Fig.~\ref{fig:QBD}, we can see that the QBD process has finite states and all states communicate with other states, so the 
Markov chain is recurrent. We also see from Fig.~\ref{fig:QBD} that every state could transition to itself, indicating that the Markov chain is aperiodic. Thus, the concerned QBD process is an ergodic Markov chain and has a unique limit state distribution $\boldsymbol{\pi}_{\omega}^{*}$ defined in (\ref{eq:pi'define}).


Notice that $\boldsymbol{\pi}_{\omega}^{*}$ must satisfy the following equation
\begin{align}
\boldsymbol{\pi}_{\omega}^{*} &= \boldsymbol{\pi}_{\omega}^{*}\mathbf{P_0}, \label{eq:pi_omega}
\end{align}
where $\mathbf{P_0}$ is the transition matrix of the QBD process, which can be determined based on (\ref{eq:QBD_matrix}) by setting the corresponding sub-matrices according to (\ref{eq:P3-B0})-(\ref{eq:P3-AM}). 
In particular,
for $M=1$ and $M=2$, the transition matrix $\mathbf{P_0}$ is given by the following (\ref{eq:matrixP0M1}) and (\ref{eq:matrixP0M2}), respectively. 
\begin{align}	\label{eq:matrixP0M1}
\bf{P_0}=& 	
\left[
					\begin{array}{cc}
									\mathbf{B_1} &	\mathbf{B_0}\\ 
									\bf{B_2} &	\bf{A_M} \\
					\end{array}
\right],
\end{align}
\begin{align}	\label{eq:matrixP0M2}
\bf{P_0}=& 	
\left[
					\begin{array}{ccc}
									\mathbf{B_1} &	\mathbf{B_0}	&		\\ 
									\bf{B_2} &	\bf{A_1}	&	\bf{A_0}	\\
									 	&			\bf{A_2}	&	\bf{A_M} \\
					\end{array}
\right].
\end{align}
Thus, under the cases of $M=1$ and $M=2$, $\boldsymbol{\pi}_{\omega}^{*}$ could be easily calculated by equations (\ref{eq:piM11})-(\ref{eq:piM13}) and (\ref{eq:piM21})-(\ref{eq:piM24}), respectively. Due to the special structure of the matrix $\mathbf{A_2}$, which is the product of a column vector $\mathbf{c}$ by a row vector $\mathbf{r}$ \cite{Latouche_Book99}, $\boldsymbol{\pi}_{\omega}^{*}$  under the case $M \geq 3$ could be calculated by equations (\ref{eq:pi'_1})-(\ref{eq:pi'_4}).  



\section{Proof of Theorem~\ref{theorem:pi'}}	\label{appendix:theorem:pi'}

Suppose that the local-queue is in some state according to the steady state distribution $\boldsymbol{\pi}^{*}_{\Omega}$, then the source delay of a packet (say $Z$) is independent of the packet generating process after $Z$ is inserted into the local-queue and is also independent of the state transitions of the local-queue after $Z$ is removed from the local-queue. Such independence makes it possible to construct a simplified QBD process to study the source delay of packet $Z$, in which new packets generated after packet $Z$ are ignored, and once $Z$ is removed from the local-queue (or equivalently the local-queue transits to state $(0,0)$), the local-queue will stay at state $(0,0)$ forever. 

For the above simplified QBD process, its transition matrix $\mathbf{P_3}$ can be determined based on (\ref{eq:QBD_matrix}) by setting the corresponding sub-matrices as follows:

For $M=1$,
\begin{align}
\mathbf{B_0} &= \mathbf{0},  \label{eq:P41-B0}\\
\mathbf{B_1} &= [1],  \label{eq:P41-B1}\\
\mathbf{B_2} &= \mathbf{c},  \label{eq:P41-B2}\\
\mathbf{A_M} &= \mathbf{Q}. \label{eq:P41-AM}
\end{align}

For $M \geq 2$,
\begin{align}
\mathbf{B_0} &= \mathbf{0},  \label{eq:P4-B0}\\
\mathbf{B_1} &= [1],  \label{eq:P4-B1}\\
\mathbf{B_2} &= \mathbf{c},  \label{eq:P4-B2}\\
\mathbf{A_0} &= \mathbf{0},  \label{eq:P4-A0}\\
\mathbf{A_1} &= \mathbf{Q},  \label{eq:P4-A1}\\
\mathbf{A_2} &= \mathbf{c}\cdot\mathbf{r}, \label{eq:P4-A2} \\
\mathbf{A_M} &= \mathbf{Q}. \label{eq:P4-AM}
\end{align}

By rearranging $\mathbf{P_3}$ as 
\begin{align}
\mathbf{P_3}	&= 														\left[
																									\begin{array}{cc}
																									1	&	\mathbf{0}	\\
																									\mathbf{c}^{+}	&	\mathbf{T}
																									\end{array}
																						\right],
\end{align}
we can see that matrices $\mathbf{c}^{+}$ and $\mathbf{T}$ are determined as (\ref{eq:cplus_M1})-(\ref{eq:T_M2}).
With matrices $\mathbf{c}^{+}$, $\mathbf{T}$ and $\boldsymbol{\pi}^{*}_{\Omega}$, the probability mass function (\ref{eq:PMF}) and CDF (\ref{eq:CDF}) of the source delay follow directly from the theory of Phase-type distribution \cite{Latouche_Book99}.

Based on the probability mass function (\ref{eq:PMF}), the mean $\overline{U}$ of the source delay can be calculated by
\begin{align}
\overline{U} &=\sum_{u=1}^{\infty}u \cdot Pr\{U=u\}	\nonumber \\
						&= \sum_{u=1}^{\infty}u \boldsymbol{\pi}_{\Omega}^{-}\mathbf{T}^{u-1}\mathbf{c}^{+}	\nonumber \\
						&=\boldsymbol{\pi}_{\Omega}^{-} \bigg(\sum_{u=1}^{\infty}u \mathbf{T}^{u-1} \bigg)\mathbf{c}^{+}	 \label{eq:matrix_series}
\end{align}

Let 
\begin{align}
f(\mathbf{T})	= \sum_{u=1}^{\infty}u \mathbf{T}^{u-1},	\label{eq:f(T)_0}
\end{align}
and use $f(x)$ to denote its corresponding numerical series
\begin{align}
f(x)	&= \sum_{u=1}^{\infty}u x^{u-1}		\\
			&= (1-x)^{-2},	\quad \text{for} \quad x < 1.	\label{eq:converge}
\end{align}

Since above simplified QBD process is actually an absorbing Markov Chain with transition matrix $\mathbf{P_3}$, we know from Theorem~$11.3$ in \cite{Charles_1997} that 
\begin{align}
\lim_{k \rightarrow \infty}\mathbf{T}^k=\mathbf{0}.	\label{eq:T_limit}
\end{align}
Based on the property (\ref{eq:T_limit}) and the Theorem~$5.6.12$ in \cite{Horn_Book90}, we can see that the spectral radius $\rho(\mathbf{T})$ of matrix $\mathbf{T}$ satisfies following condition
\begin{align}
\rho(\mathbf{T})<1.	\label{eq:radius}
\end{align}

From (\ref{eq:f(T)_0}), (\ref{eq:converge}) and (\ref{eq:radius}), it follows that the matrix series $f(\mathbf{T})$ converge as
\begin{align}
f(\mathbf{T})	&= \lim_{g \rightarrow \infty}\sum_{u=1}^{g}u \mathbf{T}^{u-1} \\
	&= (\mathbf{I}-\mathbf{T})^{-2}	\label{eq:f(T)} 
\end{align}

After substituting (\ref{eq:f(T)}) into (\ref{eq:matrix_series}), (\ref{lemma:mean_1}) then follows.

The derivation of the variance of source delay (\ref{lemma:variance})	could be conducted in a similar way and thus is omitted here.

\ifCLASSOPTIONcaptionsoff
  \newpage
\fi

\bibliographystyle{IEEEtran}
\bibliography{reference}

%

%
%
%




\end{document}